\documentclass{amsart}
\usepackage{mathrsfs,amssymb,amsmath,amsfonts,amsxtra}
\usepackage{graphicx,color,verbatim}
\usepackage[colorlinks]{hyperref}

\begin{document}
\sloppypar \sloppy
\title[The paleoclassical interpretation]{The paleoclassical interpretation of quantum theory}
\author{I. Schmelzer}
\thanks{Berlin, Germany}
\email{ilja.schmelzer@gmail.com}%
\urladdr{ilja-schmelzer.de}

\begin{abstract}
This interpretation establishes a completely classical ontology -- only the classical trajectory in configuration space -- and interprets the wave function as describing incomplete information (in form of a probability flow) about this trajectory. This combines basic ideas of de Broglie-Bohm theory and Nelsonian stochastics about the trajectory with a Bayesian interpretation of the wave function.

Various objections are considered and discussed. In particular a regularity principle for the zeros of the wave function allows to meet the Wallstrom objection. 
\end{abstract}

\maketitle

\tableofcontents

\newcommand{\pd}{\partial} 
\newcommand{\ud}{\mathrm{d}} 
\newcommand{\f}{\varphi}
\renewcommand{\a}{\alpha}
\newcommand{\w}{\nabla\times v}

\renewcommand{\H}{\mbox{$\mathcal{H}$}} 
\newcommand{\B}{\mbox{$\mathbb{Z}_2$}} 
\newcommand{\Z}{\mbox{$\mathbb{Z}$}}
\newcommand{\R}{\mbox{$\mathbb{R}$}}
\newcommand{\C}{\mbox{$\mathbb{C}$}}

\newtheorem{theorem}{Theorem}
\newtheorem{principle}{Principle}
\newtheorem{postulate}{Postulate}
\newcommand{\Sch}{Schr\"{o}dinger\/ }

\section{Introduction}

The interpretation presented here completely revives classical ontology: Reality is described completely by a classical trajectory $q(t)$, in the classical configuration space $Q$. 

The wave function is also interpreted in a completely classical way -- as a particular description of a classical probability flow, defined by the probability distribution $\rho(q)$ and average velocities $v^i(q)$. The formula which defines this connection has been proposed by Madelung 1926 \cite{Madelung} and is therefore as old as quantum theory itself. It is the polar decomposition
\begin{equation}\label{polar}
\psi(q)=\sqrt{\rho(q)} e^{\frac{i}{\hbar} S(q)},
\end{equation}
with the phase $S(q)$ being a potential of the velocity $v^i(q)$: 
\begin{equation}\label{guiding}
v^i(q) = m^{ij} \pd_j S(q),
\end{equation}
so that the flow is a potential one.\footnote{Here, $m^{ij}$ denotes a quadratic form -- a ``mass matrix'' -- on configuration space. I assume in this paper that the Hamiltonian is quadratic in the momentum variables, thus,
\begin{equation}\label{HamiltonFunction}
H(p,q) = \frac{1}{2}m^{ij}p_ip_j + V(q).
\end{equation}
This is quite sufficient for relativistic field theory, see app. \ref{relativity}.
}

This puts the interpretation into the classical realist tradition of interpretation of quantum theory, the tradition of  de Broglie-Bohm (dBB) theory \cite{deBroglie}, \cite{Bohm} and Nelsonian stochastics \cite{Nelson}. 

But there is a difference -- the probability flow is interpreted as describing incomplete information about the true trajectory $q(t)$. This puts the interpretation into another tradition -- the interpretation of probability theory as the logic of plausible reasoning in situations with incomplete information, as proposed by Jaynes \cite{Jaynes}. This objective variant of the Bayesian interpretation of probability follows the classical tradition of Laplace \cite{Laplace} (to be distinguished from the subjective variant proposed by de Finetti \cite{deFinetti}). 

So this interpretation combines two classical traditions -- classical realism about trajectories and the classical interpretation of probability as the logic of plausible reasoning. 

There is also some aspect of novelty -- a clear and certain program for development of subquantum theory. Quantum theory makes sense only as an approximation for potential probability flows. It has to be generalized to non-potential flows, described by the flow variables  $\rho(q), v^i(q)$. Without a potential $S(q)$ there will be also no wave function $\psi(q)$ in such a subquantum theory. And the flow variables are not fundamental fields themself, they also describe only incomplete information. The fundamental theory has to be one for the classical trajectory $q(t)$ alone. 

So this interpretation is also a step toward the development of a subquantum theory. This is an aspect which invalidates prejudices against quantum interpretations as pure philosophy leading nowhere. Nonetheless even this program for subquantum theory has also classical character, making subquantum theory closer to a classical theory. 

Thinking about how to name this interpretation of quantum theory, I have played around with ``neoclassical'', but rejected it, for a simple reason: There is nothing ``neo'' in it. Instead, it deserves to be named ``paleo''.\footnote{Association with the ``paleolibertarian'' political direction is welcome and not misleading -- it also revives classical libertarian ideas considered to be out of date for a long time.} 

And so I have decided to name this interpretation of quantum theory ``paleoclassical''. 

\subsection{The Bayesian character of the paleoclassical interpretation}

The interpretation of the wave function is essentially Bayesian. 

It is the objective (information-dependent) variant of Bayesian probability, as proposed by Jaynes \cite{Jaynes}, which is used here. It has to be distinguished from the subjective variant proposed by de Finetti \cite{deFinetti}, which is embraced by the Bayesian interpretation of quantum theory proposed by Caves, Fuchs and Schack \cite{Bayesian}. 

The Bayesian interpretation of the wave function is in conflict with the objective character assigned to the wave function in dBB theory. Especially Bell has emphasized that the wave function has to be understood as real object. 

But the arguments in favour of the reality of the wave function, even if strong, appear insufficient: The effective wave function of small subsystems depends on the configuration of the environment. This dependence is sufficient to explain everything which makes the wave function similar to a really existing object. It is only the wave function of a closed system, like the whole universe, which is completely interpreted in terms of incomplete information about the system itself. 

Complexity and time dependence, even if they are typical properties of real things, are characteristics of incomplete information as well. 

The Bayesian aspects essentially change the character of the interpretation: The dBB ``guiding equation'' \eqref{guiding}, instead of guiding the configuration, becomes part of the definition of the information about the configuration contained in the wave function. 

This has the useful consequence that there is no longer any ``action without reaction'' asymmetry: While it is completely natural that the real configuration does not have an influence on the information available about it, there is also no longer any ``guiding'' of the configuration by the wave function. 

\subsection{The realistic character of the paleoclassical interpretation}

On the other hand, there is also a strong classical realistic aspect of the paleoclassical interpretation: First of all, it is a realistic interpretation. But, even more, its ontology is completely classical -- the classical trajectory $q(t)$ is the only beable.  

This defines an important difference between the paleoclassical interpretation and other approaches to interpret the wave function in terms of information: It answers the question ``information about what'' by explicitly defining the ``what'' -- the classical trajectory -- and even the ``how'' -- by giving explicit formulas for the probability distribution $\rho(q)$ and the average velocity $v^i(q)$. 

The realistic character leads also to another important difference -- that of motivation. For the paleoclassical interpretation, there is no need to solve a measurement problem -- there is none already in dBB theory, and the dBB solution of this problem -- the explicit non-\Sch evolution of the conditional wave function of the measured system, defined by the wave function of system and measurement device and the actual trajectory of the measurement device -- can be used without modification.  

There is also no intention to save relativistic causality by getting rid of the non-local collapse -- the paleoclassical interpretation accepts a hidden preferred frame, which is yet another aspect of its paleoclassical character. Anyway, because of Bell's theorem, there is no realistic alternative.
\footnote{Here I, of course, have in mind the precise meaning of ``realistic'' used in Bell's theorem, instead of the metaphorical one used in ``many worlds'', which is, in my opinion, not even a well-defined interpretation.}

\subsection{The justification}

So nor the realistic approach of dBB theory, which remains influenced by the frequentist interpretation of probability (an invention of the positivist Richard von Mises \cite{Mises}) and therefore tends to objectify the wave function, nor the Bayesian approach, which embraces the anti-realistic rejection of unobservables and therefore rejects the preferred frame, are sufficiently classical to see the possibility of a completely classical picture. 

But it is one thing to see it, and possibly even to like it because of its simplicity, and another thing to consider it as justified.  

The justification of the paleoclassical interpretation presented here is based on a reformulation of classical mechanics in term of -- a wave function. It is a simple variant of Hamilton-Jacobi theory with a density added, but this funny reformulation of classical mechanics appears to be the key for the paleoclassical interpretation. The point is that its exact equivalence to classical mechanics (in the domain of its validity) and even the very fact that it shortly becomes invalid (because of caustics) almost forces us to accept -- for this classical variant -- the interpretation in terms of insufficient information. And it also provides all the details we use. 

But, then, why should we change the ontological interpretation if all what changes is the equation of motion?  Moreover, if (as it appears to be the case) the \Sch equation is simply the linear part of the classical equation, so that there is not a single new term which could invalidate the interpretation?  And where one equation is the classical limit of the other one?

We present even more evidence of the conceptual equivalence between the two equations: A shared explanation, in terms of information, that there should be a global $U(1)$ phase shift symmetry, a shared explanation of the product rule for independence, a shared explanation for homogeneity of the equation. And there is, of course, the shared Born probability interpretation. All these things being equal, why should one even think about giving the two wave functions a different interpretation? 

\subsection{Objections} 

There are a lot of objections against this interpretation to care about. 

Some of them are quite irrelevant, because they are handled already appropriately from point of view of de Broglie-Bohm theory, so I have banned them into appendices: The objection of incompatibility with relativity (app. \ref{relativity}), the Pauli objection that it destroys the symmetry between configuration and momentum variables (app. \ref{Pauli}), some doubts about viability of field-theoretic variants related with overlaps (app. \ref{fields}). 

There are the arguments in favour of interpreting the wave function as describing external real beables. These are quite strong but nonetheless insufficient: The effective wave functions of small subsystems -- and we have no access to different ones -- \emph{depend} on real beables external to the system itself, namely the trajectories of the environment. And complexity and dynamics are properties of incomplete information as well. 

And there is the Wallstrom objection \cite{Wallstrom}, the in my opinion most serious one. How to justify, in terms of $\rho(q)$ and $v^i(q)$, the ``quantization condition''
\begin{equation}
\oint m_{ij} v^i(q) dq^j = \oint \pd_j S(q) dq^j = 2\pi m\hbar, \qquad m\in\Z.
\end{equation}
for closed trajectories around zeros of the wave function, which is, in quantum theory, a trivial consequence of the wave function being uniquely defined globally, but which is completely implausible if formulated in terms of the velocity field?

Fortunately, I have found a solution of this problem in \cite{againstWallstrom}, based on an additional regularity postulate. All what has to be postulated is that \emph{$0< \Delta \rho(q) < \infty$ almost everywhere where $\rho(q)=0$}. This gives the necessary quantization condition. Moreover, there are sufficiently strong arguments that this condition can be justified by a subquantum theory. 

The idea that what we observe are localized wave packets instead of the configurations themself I reject in app. \ref{wavepackets}. So all the objections I know about can be answered in a satisfactory way. Or at least I think so. 

\subsection{Directions of future research} 

Different from many other interpretations of quantum theory, the paleoclassical interpretation suggests a quite definite program of development of a more fundamental, subquantum theory: It defines the ontology of subquantum theory as well as the equation which can hold only approximately -- the potentiality condition for the velocity $v^i(q)$. The consideration of the Wallstrom objection even identifies the domain where modifications of quantum theory are necessary -- the environment of the zeros of the wave function. 

One can identify also another domain where quantum predictions will fail in subquantum theory -- the reliability of quantum computers, in particular their ability to reach exponential speedup in comparison with classical computers. 

Another interesting question is what restrictions follow for $\rho(q), v^i(q)$ from the interpretation as a probability flow for a more fundamental theory for the trajectory $q(t)$ alone. An answer may be interesting for finding answers to the ``why the quantum'' question. A question which cannot be answered by an interpretation which is restricted to the ``what is the quantum'' question. 

\section{A wave function variant of Hamilton-Jacobi theory}

Do you know that one can reformulate classical theory in terms of a wave function?  With an equation for this wave function which is completely classical, but, nonetheless, quite close to the \Sch equation?  

In fact, this is a rather trivial consequence of the mathematics of Hamilton-Jacobi theory and the insights of Madelung \cite{Madelung}, de Broglie \cite{deBroglie}, and Bohm \cite{Bohm}. All one has to do is to look at them from another point of view. Their aim was to understand quantum theory, by representing quantum theory in a known, more comprehensible, classical form -- a form remembering the classical Hamilton-Jacobi equation
\begin{equation}\label{HamiltonJacobi}
\pd_t S(q) + \frac12 m^{ij} \pd_i S(q) \pd_j S(q) + V(q) = 0.
\end{equation}
So, the real part of the \Sch equation, divided by the wave function, gives
\begin{equation}\label{Bohm}
\pd_t S(q) + \frac12 m^{ij} \pd_i S(q) \pd_j S(q) + V(q) + Q[\rho] = 0,
\end{equation}
with only one additional term -- the quantum potential 
\begin{equation}\label{Qdef}
Q[\rho] = -\frac{\hbar^2}{2} \frac{m^{ij} \pd_i\pd_j \sqrt{\rho}}{\sqrt{\rho}}
\end{equation}
The imaginary part of the \Sch equation (also divided by $\psi(q)$) is the continuity equation for $\rho$
\begin{equation}\label{continuity}
\pd_t\rho(q,t) + \pd_i(\rho(q,t)v^i(q,t)) = 0
\end{equation}
 
Now, all what we have to do is to revert the aim -- instead of presenting the \Sch equation like a classical equation, let's present the classical Hamilton-Jacobi equation like a \Sch equation. There is almost nothing to do -- to add a density $\rho(q)$ together with a continuity equation \eqref{continuity} is trivial.  We use the same polar decomposition formula \eqref{polar} to define the classical wave function. It remains to do the same procedure in the other direction. The difference is the same -- the quantum potential. So we obtain, as the equation for the classical wave function, an equation I will name here pre-\Sch equation:
\begin{equation}\label{pre}
i\hbar \pd_t \psi(q,t) = -\frac{\hbar^2}{2} m^{ij}\pd_i\pd_j\psi(q,t) + (V(q) - Q[\rho]) \psi(q,t) = \hat{H} \psi - Q[|\psi|^2]\psi.
\end{equation}

Of course, the additional term is a nasty, nonlinear one. But that's the really funny point: We can obtain now quantum theory as a linear approximation of classical theory, in other words, as a simplification\footnote{But note the important point that the very definition of ``linear approximation'' depends on the definition of the ``wave function.'', which depends on a choice of the classically arbitrary parameter $\hbar$. So there is not, as usual, a meaningful notion of ``\emph{the} linear approximation''.}

But there is more than this in this classical equation for the classical wave equation. The point is that there is an exact equivalence between the different formulations of classical theory, despite the fact that one is based on a classical trajectory $q(t)$ only, and the other has a wave function $\psi(q,t)$ together with the trajectory $q(t)$. And it is this exact equivalence which can be used to identify the meaning of the classical wave function. 

Because of its importance, let's formulate this in form of a theorem:
\begin{theorem}[equivalence]
Assume $\psi(q,t),q(t)$ fulfill the pre-\Sch equation \eqref{pre} for some Hamilton function $H(p,q)$ of type \eqref{HamiltonFunction}, together with the guiding equation \eqref{guiding}. 

Then, whatever the initial values $\psi_0(q),q_0$ for the wave function $\psi_0(q)=\psi(q,t_0)$ and the initial configuration $q_0=q(t_0)$, there exists initial values $q_0,p_0$ so that the Hamilton equation for $H(p,q)$ with these initial values gives $q(t)$.  
\end{theorem}
\begin{proof} 
The difficult part of this theorem is classical Hamilton-Jacobi theory, so, presupposed here as known. The simple modifications of this theory used here -- adding the density $\rho(q)$, with continuity equation \eqref{continuity}, and rewriting the result in terms of the wave function $\psi(q)$ defined by the polar decomposition \eqref{polar}, copying what has been done by Bohm \cite{Bohm} -- do not endanger the equivalence between Hamiltonian (or Lagrangian) and Hamilton-Jacobi theory. 
\end{proof}

\section{The wave function describes incomplete information}

So let's evaluate now what follows from the equivalence theorem about the physical meaning of the (yet classical) wave function. 

First, it is the classical (Hamiltonian or Lagrangian) variant which is preferable as a fundamental theory. There are three arguments to justify this:

\begin{itemize}

\item Simplicity of the set of fundamental variables: We need only a single trajectory $q(t)$ instead of trajectory $q(t)$ together with a wave function $\psi(q,t)$.  

\item Correspondence between fundamental variables and observables: It is only the trajectory $q(t)$ in configuration space which is observable. 

\item Stability in time: The wave function develops caustics after a short period of time and becomes invalid. The classical equations of motion does not have such a problem. 

\end{itemize}

So it is the classical variant which is preferable as a fundamental theory. Thus, we can identify the true beable of classical theory with the classical trajectory $q(t)$. 

The next observation is that, once $q(t)$ is known and fixed, the wave function contains many degrees of freedom which are unobservable in principle: Many different wave functions define the same trajectory $q(t)$. So we can conclude that these different wave functions do \emph{not} describe physically different states, containing additional beables. 

On the other hand, this is true only if $q$ is known. What if $q$ is not known? In this case, the wave function defines simply a subset of all possible classical trajectories, and a probability measure on this subset.  

To illustrate this, a particular most important example of a Hamilton-Jacobi function is useful: It is the function $S(q_0,t_0,q_1,t_1)$ defined by
\begin{equation}\label{Minimumproblem}
 S(q_0,t_0,q_1,t_1) = \int_{t_0}^{t_1} L(q(t),\dot{q}(t),t) dt,
\end{equation} 
where the integral is taken over the classical solutions $q(t)$ of this minimum problem with initial and final values $q(t_0)=q_0$, $q(t_1)=q_1$. This function fulfills the Hamilton-Jacobi equation in the variables $q_0,t_0$ as well as $q_1,t_1$. In both versions, it can be characterized as a Hamilton-Jacobi function $S(q,t)$ which is defined by a subset of trajectories: The function $S(q_0,t_0,q,t)$ describes the subset of trajectories going through $q_0$ at $t_0$, while the function $S(q,t,q_1,t_1)$ describes the subset of trajectories going through $q_1$ at $t_1$.  

We can generalize this. The phase $S(q,t)$ tells us the value of $\dot{q}(t)$ given $q(t)$: \emph{If} $q(t)=q_0$ \emph{then} $\dot{q}(t)=v_0$, with $v_0$  defined by the guiding equation \eqref{guiding}. So, $S(q,t)$ always distinguishes a particular subset of classical trajectories. 

Even more specific, this subset described by $S(q,t)$ can be uniquely described in terms of a subset of the possible initial values at the moment of time $t$ -- the configuration $q(t_0)$ and the momentum $p(t_0)=\nabla S(q(t_0),t_0)$, that means, as a subset of the possible values for the fundamental beables $q,p$ (or $q,\dot{q}$). 

The other part -- the density $\rho(q)$ -- is nothing but a probability density on this particular subset. 

Of course, a subset is nothing but a special delta-like probability measure, so that the wave function simply defines a probability density on the set of possible initial values:
\begin{equation}\label{probability}
\rho(p,q)dpdq = \rho(q)\delta(p-\nabla S(q))dpdq
\end{equation}
The pre-\Sch equation is, therefore, nothing but the evolution equation for this particular probability distribution. 

So our classical, Hamilton-Jacobi wave function is nothing mystical, but simply a very particular form of a standard classical probability density $\rho(p,q)$ on the phase space. 

In particular, the pre-\Sch equation for the wave function is nothing but a particular case of the Liouville equation, the standard law of evolution of standard classical probability distributions $\rho(p,q)$, for the particular ansatz \eqref{probability}, and follows from the fundamental law of evolution for the true, fundamental classical beables. 

Moreover, the Liouville equation also defines the domain of applicability of the equation for the wave function. This domain is, in fact, restricted. In terms of $\rho(p,q)$, it will always remain a probability distribution on some Lagrangian submanifold. But this Lagrangian submanifold will be, after some time, no longer a graphic of a function $p=\nabla S(q)$ on configuration space -- there may be caustics, and in this case there will be several values of momentum for the same configuration $q$. If this happens, the wave function is no longer an adequate description. 

Such an effect -- restricted validity -- is quite natural for the evolution of information, but not for fundamental beables. 

So the wave function variant of Hamilton-Jacobi theory almost forces us to accept an interpretation of the wave function in terms of incomplete information. Indeed,

\begin{itemize}

\item The parts of the wave function, $\rho(q)$ as well as $S(q)$, make sense as describing a well-defined type of incomplete information about the classical configuration, namely the probability distribution $\rho(p,q)$ defined by \eqref{probability}.  

\item The alternative, to interpret the wave function as describing some other, external beables, does not make sense, given the observational equivalence of the theory with simple Lagrangian classical mechanics, with $q(t)$ as the only observable. Additional beables should influence, at least in some circumstances, the $q(t)$. They don't.  

\end{itemize}

So it looks like a incomplete information, it behaves like incomplete information, it becomes invalid like incomplete information -- it is incomplete information.

And, given that we know, in the case of the pre-\Sch equation, the fundamental law of evolution of the beables themself, it also makes no sense to reify this particular probability distribution as objective.  A probability distribution defines a set of incomplete information about the real configuration, that's all. 

\subsection{What about the more general case?}

Of course, the considerations both have been based on the equivalence theorem between the classical evolution and the evolution defined by the pre-\Sch equation. It was this exact equivalence which was able to give us some certainty, to remove all doubts that there is something else, some information about other beables, hidden in the wave function. 

But what about the more general case, the case where we do not have an exact equivalence between an equation in terms of $q,\psi(q)$ and an equation purely in terms of the classical trajectory $q(t)$?  In such a situation, the case for an interpretation of $\psi(q)$ in terms of incomplete information about the $q(t)$ is, of course, a little bit weaker. 

Nonetheless, given such an ideal correspondence for the pre-\Sch equation, the interpretation remains certainly extremely plausible in a more general situation too. Indeed, why should one change the interpretation, the ontological meaning, of $\psi(q)$, if all what is changed is that we have replaced the pre-\Sch equation by another equation? The evolution equation is different, that's all. In itself, a modification of the evolution equation does not give even a bit of motivation to change the interpretation. 

And I think that to find such a motivation is a really hard job. The similarity between the two variants of the \Sch equation does not makes it easier: The funny observation that the \Sch equation is the linear part of the pre-\Sch equation becomes relevant here.  If the \Sch equation would contain some new terms, this would open the door for attempts to show that the new terms do not make sense in the original interpretation. But there are no new terms in the linear part of the pre-\Sch equation. All terms of the \Sch equation are already contained in the pre-\Sch equation. So they all make sense.\footnote{It has to be mentioned in this connection that there is something present in the \Sch equation which is not present in the pre-\Sch equation -- the dependence on $\hbar$. So one can at least try to base an argument on this additional dependence. But, as described below, if one incorporates a Nelsonian stochastic process, the dependence on $\hbar$ appears in a natural way as connected with the stochastic process. So to use this difference to justify a modification of the ontology remains quite nontrivial.}

\subsection{The classical limit}

There is also another strong argument for interpreting both theories in the same way -- that classical theory appears as the classical limit of quantum theory.  

The immediate consequence of this is that both theories have the same intention -- the description, as accurate as possible, of the same reality.  

So this is not a situation where the same mathematical apparatus is applied to quite different phenomena, so that it is natural use a different interpretation even if the same mathematical apparatus is applied. In our case, we use the same mathematical formalism to describe the same thing. At least in the classical limit, quantum theory has to describe the same thing -- with the same interpretation -- as the classical interpretation. 

\section{What follows from the interpretation}

But there is not only the similarity between the equations, and that the object is essentially the same, which suggests to use the same interpretation for both variants of the wave function. 

There are also some interesting consequences of the interpretation. And these consequences, to be shared by all wave functions which follow this interpretation, are fulfilled by the quantum wave function.  

\subsection{The relevance of the information contained in the wave function} 

A first consequence of the interpretation can be found considering the \emph{relevance} of the information contained in the wave function, as information about the real beables. In fact, assume that the wave function $\psi(q)$ contains a lot of information which is irrelevant as information about the $q$. This would open the door for some speculation about the nature of this additional information. Once it cannot be interpreted as information about the $q$, it has to contain information about some other real beables. So let's consider which of the information contained in the wave function is really relevant, tells us something about the true trajectory $q(t)$.

Now, knowledge about the probability $\rho(q)$ is certainly relevant if we don't know the real value of $q$. And it is relevant in all of its parts. 

Then, as we have seen, $S(q)$ gives us, via the ``guiding equation'' \eqref{guiding}, the clearly relevant information about the value of $\dot{q}$ given the value of $q$. Given that we don't know the true value of $q$, this information is clearly relevant. But, different from $\rho(q)$, the function $S(q)$ contains also a little bit more: Another function $S'=S+c$ would give \emph{exactly} the same information about the real beables. 

As a consequence, the wave function $\psi(q)$ also contains a corresponding additional, irrelevant information. The polar decomposition \eqref{polar} defines what it is  -- a global constant phase factor.  

So we find that all the information contained in $\psi(q)$ -- \emph{except for a global constant phase factor} -- is relevant information about the real beables $q$. 

At a first look, this seems trivial, but I think it is a really remarkable property of the paleoclassical interpretation. 

The irrelevance of the phase factor is a property of the interpretation of the meaning of the wave function. It follows from this interpretation: We have considered all the ways how to obtain information about the beables from $\psi(q)$, and we have found that all the information is relevant, except this constant phase factor. 

And now let's consider the equations. Above equations considered up to now -- the pre-\Sch equation as well as the \Sch equation -- have the same symmetry regarding multiplication with a constant phase factor: Together with $\psi(q,t)$, $c\psi(q,t)$ is a solution of the equations too.  

This is, of course, how it should be if the wave function describes incomplete information about the $q$. If, instead, it would describe some different, external beables, then there would be no reason at all to assume such a symmetry. In fact, why should the values of some function, values at points far away from each other, be connected with each other by such a global symmetry?  

So every interpretation which assigns the status of reality to $\psi(q)$ has, in comparison, a problem to explain this symmetry. A popular solution is to assign reality not to $\psi(q)$, but, instead, to the even more complicate density matrix $|\psi\rangle\langle\psi|$. The flow variables are, obviously, preferable by simplicity.

\subsection{The independence condition}

It is one of the great advantages of Jaynes' information-based approach to plausible reasoning \cite{Jaynes} that it contains common sense principles to be applied in situations with insufficient information. If we have no information which distinguishes the probability of the six possible outcomes of throwing a die, we \emph{have} to assign equal probability to them. Everything else would be irrational.\footnote{The purely subjectivist, de Finetti approach differs here. It does not make prescriptions about the initial probability distributions -- all what is required is that one updates the probabilities following the rules of Bayesian updating. This difference is one of the reasons for my preference for the objective approach proposed by Jaynes \cite{Jaynes}}  

Similar principles work if we consider the question of independence. From a frequentist's point of view, independence is a quite nontrivial physical assumption, which has to be tested. And, in fact, there is no justification for it, at least none coming from the frequentist interpretation. From point of view of the logic of plausible reasoning, the situation is much better: Once we have no \emph{information} which justifies any hypothesis of a dependence between two statements $A$ and $B$, we \emph{have} to assume their independence. There is no information which makes $A$ more or less plausible given $B$, so we have to assign equal plausibilities to them: $P(A|B)=P(A)$.  But this is the condition of independence $P(AB)=P(A)P(B)$. 

The different status of plausibilities in comparison with physical hypotheses makes this obligatory character reasonable. Probabilities are not hypotheses about reality, but logical conclusions derived from the available information, derived using the logic of plausible reasoning.

Now, all this is much better explained in \cite{Jaynes}, so why I talk here about this?  The point is that I want to obtain here a formula which appropriately describes different independent subsystems. Subsystems are independent if we have no information suggesting their dependence.  A quite typical situation, so independence is quite typical too.  

So assume we have two subsystems, and we have no information suggesting any dependence between them. What do we have to assume based on the logic of plausible reasoning? 

For the probability distribution itself, the answer is trivial:  
\begin{equation}
\rho(q_1,q_2)=\rho_1(q_1)\rho_2(q_2).
\end{equation}
But what about the phase function $S(q)$?  We have to assume that the velocity of one system is independent of the state of the other one. So the potential of that velocity also should not depend on the state of the other system. So we obtain 
\begin{equation}
S(q_1,q_2) = S_1(q_1) + S_2(q_2).
\end{equation}
This gives, combined using the polar decomposition \eqref{polar}, the following rule for the wave function:
\begin{equation}
\psi(q_1,q_2) = \psi_1(q_1) \psi_2(q_2). 
\end{equation}

Of course, again, the rule is trivial. Nobody would have proposed any other rule for defining a wave function in case of independence.  Nonetheless, I consider it as remarkable that it has been derived from the very logic of plausible reasoning and the interpretation of the wave function. 

And, again, this property is shared by both variants, the quantum as well as the classical one. And if one thinks about the applicability of this rule to the classical variant of the wave function, one has to recognize that it is nontrivial: From a classical point of view, the rule of combining $\rho(q)$ and $S(q)$ into a wave function $\psi(q)$ is quite arbitrary, and there is no a priori reason to expect that such an arbitrary combination transforms the rule for defining independent classical states into a simple rule for independent wave functions, moreover, into the one used in quantum theory. 

\subsection{Appropriate reduction to a subsystem}

While the pre-\Sch equation is non-linear, it shares with the \Sch equation the weaker property of homogeneity of degree one: The operator condition 
\begin{equation}\label{homogenity}
\Omega(c\psi) = c \Omega(\psi)
\end{equation}
holds not only for the \Sch operator, but also for the non-linear pre-\Sch operator, and not only for $|c|=1$, as required by the $U(1)$ symmetry, which we have already justified, but for arbitrary $c$, so that the $U(1)$ symmetry is, in fact, a $GL(1)$ symmetry of the equations. 

So one ingredient of linearity is shared by the pre-\Sch equation. This suggests that for this property there should be a more fundamental explanation, one which is shared by both equations. And, indeed, such an explanation is possible. There should be a principle which allows an appropriate reduction of the equation to subsystems, something like the following

\begin{principle}[splitting principle]
There should a be simple ``no interaction'' conditions for the operators on a system consisting of two subsystems of type $\Omega=\Omega_1+\Omega_2$ such that the equation $\pd_t\psi=\Omega(\psi)$ splits for independent product states $\psi(q_1,q_2)=\psi_1(q_1)\psi_2(q_2)$ into independent equation's for the two subsystems. 
\end{principle}

Now, the time derivative splits nicely:
\begin{equation}
\pd_t \psi(q_1,q_2)= \pd_t\psi_1(q_1)\psi_2(q_2) + \psi_1(q_1)\pd_t\psi_2(q_2).
\end{equation}
It remains to insert the equations for the whole system $\pd_t \psi = \Omega(\psi)$ as well as for the two subsystems into this equation. This gives:
\begin{equation}
\Omega(\psi_1(q_1)\psi_2(q_2)) = \Omega_1(\psi_1(q_1))\psi_2(q_2) + \psi_1(q_1)\Omega_2(\psi_2(q_2)),
\end{equation}
where the $\Omega_i$ are subsystem operators acting only on the $q_i$, so that functions of the other variable are simply constants for them. On the other hand, in the splitting principle we have assumed $\Omega = \Omega_1+\Omega_2$. Comparison gives
\begin{equation}
\begin{split}
\Omega_1(\psi_1(q_1)\psi_2(q_2)) &= \Omega_1(\psi_1(q_1))\psi_2(q_2),\\
\Omega_2(\psi_1(q_1)\psi_2(q_2)) &= \psi_1(q_1)\Omega_2(\psi_2(q_2)).
\end{split}
\end{equation}
This follows from the homogeneity condition \eqref{homogenity}. The weaker $U(1)$ symmetry is not sufficient, because the values of $\psi_2(q_2)$ may be arbitrary complex numbers. 

Something similar to the splitting property is necessary for any equation relevant for us -- we have not enough information to consider the equation of the whole universe, thus, have to restrict ourself to small subsystems. And the equations of these subsystems should be at least approximately independent of the state of the environment. 

So the homogeneity of the equations -- in itself a quite nontrivial property of the equations, given the definition of the wave function by polar decomposition -- can be explained by this splitting property. 

\subsection{Summary}

So we have found three points, each in itself quite trivial, but each containing some nontrivial element of explanation based on the paleoclassical principles: The global $U(1)$ phase shift symmetry of the wave function, explained by the irrelevance of the phase as information about $q(t)$, the product rule of independence, explained by the logic of plausible reasoning, applied to the set of information described by the wave function, and the homogeneity of the equations, explained in terms of a splitting property for independent equations for subsystems. 

All three points follow the same scheme -- the interpretation in terms of incomplete information allows to derive a rule, and this rule is shared by both theories, quantum theory as well as the wave function variant of classical theory.  

So all three points give additional evidence that the simple, straightforward proposal to use the same interpretation of the wave function for both theories is the correct one.

\section{Incorporating Nelsonian stochastics}

The analogy between pre-\Sch and \Sch equation is a useful one, but it is that useful especially because the equations are nonetheless very different. And one should not ignore these differences. 

One should not, in particular, try to interpret the \Sch equation as \emph{the} linear approximation of the pre-\Sch equation: The pre-\Sch equation does not depend on $\hbar$. Real physical effects do depend. And so there should be also something in the fundamental theory, the theory in terms of the trajectory $q(t)$, which depends on $\hbar$. 

Here, Nelsonian stochastics comes into mind. In Nelsonian stochastics the development of the configuration $q$ in time is described by a deterministic drift term $b^i(q(t),t)dt$ and a stochastic diffusion term $dB^i_t$
\begin{equation}\label{process}
dq^i(t) = b^i(q(t),t) dt + dB^i_t,
\end{equation}
with  is a classical Wiener process $B^i_t$ with expectation $0$ and variance
\begin{equation}
\langle dB^i_t m_{ij} dB^j_t \rangle = \hbar dt,
\end{equation}
so that we have a $\hbar$-dependence on the fundamental level. The probability distribution $\rho(q(t),t)$ has, then, to fulfill the Fokker-Planck-equation:
\begin{equation}\label{FokkerPlanck}
\partial_t \rho + \pd_i (\rho b^i) - \frac{\hbar}{2}m^{ij}\pd_i\pd_j\rho = 0
\end{equation}
For the average velocity $v^i$ one obtains
\begin{equation}\label{vdef}
v^i(q(t),t) = b^i(q(t),t) - \frac{\hbar}{2}\frac{m^{ij}\pd_j \rho(q(t),t)}{\rho(q(t),t)},
\end{equation}
which fulfills the continuity equation. The difference between flow velocity $b^i$ and average velocity $v^i$, the osmotic velocity
\begin{equation}\label{osmotic}
u^i(q(t),t) = b^i(q(t),t)-v^i(q(t),t) = \frac{\hbar}{2}\frac{m^{ij}\pd_j\rho(q(t),t)}{\rho(q(t),t)} = \frac{\hbar}{2}m^{ij}\pd_j \ln\rho(q(t),t)
\end{equation}
has a potential $\ln\rho(q)$. The average acceleration is given by
\begin{equation}\label{acceleration}
a^i(q(t),t) = \partial_t v^i + v^j \pd_j v^i - \frac{\hbar^2}{2} m^{ij}\pd_j
        \left(\frac{m^{kl}\pd_k\pd_l \sqrt{\rho(q(t),t)}}{\sqrt{\rho(q(t),t)}}\right)
\end{equation}
For this average acceleration the classical Newtonian law
\begin{equation}\label{Newton}
a^i(q(t),t) = - m^{ij} \pd_j V(q(t),t)
\end{equation}
is postulated. Putting this into the equation \eqref{acceleration} gives
\begin{equation}\label{derivativeBohm}
\partial_t v^i + (v^j \pd_j)v^i =
 - m^{ij} \pd_j \left(V -\frac{\hbar^2}{2}
        \frac{\Delta \sqrt{\rho}}{\sqrt{\rho}}
\right) = - m^{ij}\pd_j(V + Q[\rho]).
\end{equation}

The next postulate is that the average velocity $v^i(q)$ has, at places where $\rho>0$, a potential, that means, a function $S(q)$, so that \eqref{guiding} holds. Then equation \eqref{derivativeBohm} can be integrated (putting the integration constant into $V(q)$), which gives \eqref{Bohm}. Finally, one combines the equations for $S(q)$ and $\rho(q)$ into the \Sch equation as in de Broglie-Bohm theory. 

So far the basic formulas of Nelsonian stochastics. Now, the beables of Nelsonian stochastics are quite different from those of the paleoclassical interpretation. The external flow $b^i(q,t)$ does not define some state of information, but some objective flow which takes away the configuration if it is at $q$. The configuration is guided in a way close to a swimmer in the ocean: He can randomly swim into one or another direction, but wherever he is, he is always driven by the flow of the ocean, a flow which exists independent of the swimmer. 

What would be the picture suggested for a stochastic process by the paleoclassical interpretation? It would be different, more close to a spaceship flying in the cosmos. If it changes its place because of some stochastic process, there will be no predefined, independent velocity $b^i(q)$ at the new place which can guide it. Instead, it would have to follow its previous velocity. Indeed, the velocity field $v^i(q)$, and, therefore, also $b^i(q)$, of the new place describes only information, not a physical field which could possibly influence the spaceship in its new position. 

But what follows for the information if there is some probability that the stochastic process causes the particle to change its location to the new one? It means that the old average velocity has to be recomputed. 

Of course, there is such a necessity to recompute only if the velocity at the new location is different from the old one.  So, if $\pd_i v^k=m^{jk}\pd_i p_j=0$, there would be no need for such a recomputation at all. This condition is clearly too strong, it would cause the whole velocity field to become constant. Now there is an interesting sub-condition, namely that $\pd_i p_j-\pd_j p_i=0$, the condition that the velocity field is a potential one. So the re-averaging of the average velocity $v^i(q)$ caused by the stochastic process will decrease the curl. 

This gives a first advantage in comparison with Nelsonian stochastics: The potentiality assumption does not have to be postulated, without any justification, for the external flow $b^i(q,t)$. (It is postulated for average velocity, but their difference -- the osmotic velocity -- has a potential anyway.) In the paleoclassical interpretation we have an independent motivation for postulating this. But, let's recognize, there is no fundamental reason to postulate potentiality: On the fundamental level, the curl may be nontrivial. All what is reached is that the assumption of potentiality is not completely unmotivated, but that it appears as a natural consequence of the necessary re-averaging.

There is another strangeness connected with the external flow picture of Nelsonian stochastics. The probability distribution $\rho(q)$ already characterizes the information about the configuration itself -- it is a probability for the swimmer in the ocean, not for the ocean. Once they depend on $\rho(q)$, as the average velocity $v^i(q)$, as the average acceleration $a^i(q)$ also describe information about the swimmer, not about the ocean. Then, it is postulated that the average acceleration of the swimmer has to be defined by the Newtonian law. But because this average velocity is, essentially, already defined by the very process -- the flow of the ocean -- together with the initial value for the swimmer, this condition for the swimmer becomes an equation for the ocean. This is conceptually unsound -- as if the ocean has to care that the swimmer is always correctly accelerated.  

But this conceptual inconsistency disappears in the paleoclassical interpretation. The drift field is now part of the incomplete information about the configuration itself, as defined by the average velocity and osmotic velocity. There is no longer any external drift field. And, so, it is quite natural that a condition for the average acceleration of the configuration gives an equation for the average velocity $v^i(q)$. So the paleoclassical picture is internally much more consistent.

\section{Enthropic dynamics}

The inconsistency of Nelsonian stochastics we have mentioned above has been also solved in another approach -- the ``entropic dynamics'' \cite{Caticha}. In this approach, the external flow $b^i(q,t)$ is a potential flow, with the entropy $S(q)$ as the potential of this flow. 

Conceptually, this approach is very close to this approach. We have a configuration $q$ (denoted there by $x$ to describe the position of one or more non-relativistic particles) together with an unspecified set of additional variables $y$. These other variables somehow interact with $q$, and their probability distribution $p(y|q)$ depends on $q$ and motivates the definition of a corresponding entropy $S(q) = \int p(y|q) \ln p(y|q) dy$ depending on $q$ too. 

How can this be made compatible with the paleoclassical approach, which does not have such external degrees of freedom?  This appears quite easy: We always consider only some small subsystem of the whole universe. So, we always have two parts of the whole universe: first, the system $q$ we study, and, then, the remaining part of the universe, the environment $q_{env}$, so that $q_{univ} = (q, q_{env})$.

Unfortunately, another point seems more problematic: The condition that the probability flow has a potential obtains a more fundamental role in this interpretation. Here, the potential of the flow is identified with the entropy. This makes the potential a well-defined global function, which is defined (as it would look like in our approach) by $S(q) = \int p(q_{env}|q) \ln p(q_{env}|q) dq_{env}$. 

This makes it more problematic to solve the Wallstrom objection. The point is that our proposal for the solution of the Wallstrom objection requires that the probability flow does not have a global potential: It can have a potential locally, but not near the zeros of the wave function, thus, the local potentials cannot be used to define a global potential.  

But the situation is not at all hopeless. The identification of $S(q)$ with entropy (as a well-defined global function) holds only for a part of the probability flow, the one which corresponds to the external flow $b^i(q,t)$. There is also the other part, the osmotic flow. Of course, the osmotic velocity has also a potential $\ln\rho(q)$. But this potential becomes anyway infinite near the zeros of the wave function, thus, needs to be regularized in the entropic approach too. This regularization can possibly solve the problem, following the lines proposed in sect. \ref{Wallstrom} and \cite{Wallstrom}.

\section{The character of the wave function}

Let's start with the consideration of the objections against the paleoclassical interpretation. Given that the basic formulas are not new at all, I do not have to wait for reactions to this paper -- some quite strong arguments are already well-known.

The first one is the evidence in favour of the thesis that the wave function describes the behaviour of real degrees of freedom, degrees of freedom which actually influence the things we can observe immediately. Here, at first Bell's argumentation comes into mind -- an argumentation for a double ontology which, I think, has impressed many of those who support today realistic interpretations: 
\begin{quote}
Is it not clear fro the smallness of the scintillation on the screen that we have to do with a particle? And is it not clear, from the diffraction and interference patterns, that the motion of the particle is directed by a wave? 
(\cite{Bell} p. 191).\end{quote}
But what are the points which make this argument that impressive? What is it what motivates us to accept some things as real? Here I see no way to express this better than Brown and Wallace: 
\begin{quote}
From the corpuscles' perspective, the wave-function is just a (time-dependent) function on their configuration space, telling them how to behave; it superficially appears similar to the Newtonian or Coulomb potential field, which is again a function on configuration space. No-one was tempted to reify the Newtonian potential; why, then, reify the wave-function?

Because the wave-function is a very different sort of entity. It is contingent (equivalently, it has dynamical degrees of freedom independent of the corpuscles); it evolves over time; it is structurally overwhelmingly more complex (the Newtonian potential can be written in closed form in a line; there is not the slightest possibility of writing a closed form for the wave-function of the Universe.) Historically, it was exactly when the gravitational and electric fields began to be attributed independent dynamics and degrees of freedom that they were reified: the Coulomb or Newtonian `fields' may be convenient mathematical fictions, but the Maxwell field and the dynamical spacetime metric are almost universally accepted as part of the ontology of modern physics.

We don't pretend to offer a systematic theory of which mathematical entities in physical theories should be reified. But we do claim that the decision is not to be made by fiat, and that some combination of contingency, complexity and time evolution seems to be a requirement.
(\cite{BrownWallace} p. 12-13)\end{quote}

So, let's consider the various points in favour of the reality of the wave function:

\subsection{Complexity of the wave function} 

The argument of complexity seems powerful. But, in fact, for an interpretation in terms of incomplete information \emph{this} is not a problem at all. 

Complexity is, in fact, a natural consequence of incompleteness of information.  The complete information about the truth of a statement is a single bit:  true or false.  The incomplete information is much more complex: it is a real number, the probability.  

As well, the complete information about reality in this interpretation is simple: a single trajectory $q(t)$.  But incomplete information requires much more information: Essentially, we need probabilities for all possible trajectories.  

\subsection{Time dependence of the wave function} 

Time dependence is, as well, a natural property of information -- complete or incomplete.  The information about where the particle has been yesterday transforms into some other information about where the particle is now. 

This transformation is, moreover, quite nontrivial and complex. 

It is also worth to note here that the law of transformation of information is a derivative of the real, physical law of the behaviour of the real beables.  

So it necessarily has all the properties of a physical law. 

We can, in particular, use the standard Popperian scientific method (making hypotheses, deriving predictions from them, testing and falsifying them, and inventing better hypotheses) to find these laws. 

This is, conceptually, a quite interesting point: The laws of probability themself are best understood, following Jaynes \cite{Jaynes}, as laws of extended logic, of the logic of plausible reasoning. 

But, instead, the laws of \emph{transformation} of probabilities in time follow from the laws of the original beables in time, and, therefore, have the character of physical laws.

Or, in other words, incomplete information develops in time in a way indistinguishable from the development in time of real beables. In particular, we use the same objective scientific methods to find and test them. 

So, nor the simple fact that there is a nontrivial time evolution, nor the very physical character of the dynamical laws, in all details, up to the scientific method we use to find them, give any argument against an interpretation in terms of incomplete information.  All this is quite natural for the evolution of incomplete information too. 

\subsection{The contingent character of the wave function}

There remains the most powerful argument in favour of the reality of the wave function: Its contingent character.  

There are different wave functions, and these different wave functions lead to objectively different probability distributions of the observable results.  

If we have, in particular, different preparation procedures, leading to different interference pictures, we really observe different interference pictures. It is completely implausible that these different interference pictures -- quite objective pictures -- could be the result of different sets of incomplete information about the same reality. The different interference pictures are, clearly, the result of different things happening in reality.  

But, fortunately, there is not even a reason to disagree with this. The very point is that one has to distinguish the wave function of a small subsystem -- we have no access to any other wave functions -- from the wave function of a closed system. The last is, in fact, only an object of purely theoretical speculation, because there is no closed system in nature except the whole universe, but we have no idea about the wave function of the whole universe. 

For the wave function of a small subsystem, the situation is quite different. It does not contain only incomplete information about the subsystem. In fact, it is only an effective wave function, and there is a nice formula of dBB theory, which can be used in our paleoclassical interpretation too: The formula which defines the conditional wave function of a subsystem $\psi^S(q_S,t)$ in terms of the wave function of the whole system (say the whole universe) $\psi(q_S,q_E,t)$ and the configuration of the environment $q_E(t)$:
\begin{equation}\label{conditional}
\psi^S(q_S,t) = \psi(q_S,q_E(t),t)
\end{equation}
This is a remarkable formula of dBB theory which contains, in particular, the solution of the measurement problem: The evolution of $\psi^S(q_S,t)$ is, in general, not described by a \Sch equation -- if there is interaction with the environment, the evolution of $\psi^S(q_S,t)$ is different, but, nonetheless, completely well-defined.  And this different evolution is the collapse of the wave function caused by measurement. 

Let's note that the paleoclassical interpretation requires to justify this formula in terms of the information about the subsystem. But this is not a problem. Indeed, assume the trajectory of the environment $q_E(t)$ is known -- say by observation of a classical, macroscopic measurement device. Then the combination of the knowledge described by the wave function of the whole system with the knowledge of $q_E(t)$ gives exactly the same knowledge as that described by $\psi^S(q_S,t)$. Indeed, the probability distribution gives
\begin{equation}
\rho^S(q_S,t) = \rho(q_S,q_E(t),t),
\end{equation}
and, similarly, the velocity field defined by $S(q)$ follows the same reduction principle:
\begin{equation}
\nabla S^S(q_S,t) = \nabla S(q_S, q_E(t),t).
\end{equation}
So in the paleoclassical interpretation the dBB formula the conditional wave function of a subsystem is a logical necessity. This provides yet another consistency check for the interpretation. 

But the reason for considering this formula here was a different one: The point is that the wave function of the subsystem in fact contains important information about other real beables -- the actual configuration of the whole environment $q_E(t)$.  So there are real degrees of freedom, different from the configuration of the system $q_S(t)$ itself, which are distinguished by different wave functions $\psi^S(q_S,t)$. 

And we do not have to object at all if one argues that the wave function contains such additional degrees of freedom. That's fine, it really contains them. These degrees of freedom are those of the configuration of the environment $q_E(t)$.  And this is not an excuse, but a logical consequence of the interpretation itself, a consequence of the definition of the conditional wave function of the subsystem \eqref{conditional}.

\subsection{The PBR theorem}

The formula \eqref{conditional} is also the key why we do not have to bother about impossibility theorems like the PBR theorem \cite{PBR}. In this theorem, one considers two non-orthogonal states, in the simplest case $|\psi_0\rangle = |0\rangle$ and $|\psi_1\rangle=|+\rangle = (|0\rangle+|1\rangle)/\sqrt{2}$. These states represent some probability distributions $\mu_0(\lambda)$ resp.  $\mu_1(\lambda)$ over some space of possible states of reality $\Lambda$. If the wave function is not completely ontological, but at least in part epistemological, then it will happen that the corresponding probability distributions $\mu_0(\lambda), \mu_1(\lambda)$ have a nontrivial overlap -- this would be the case when the difference between $|\psi_0\rangle$ and $|\psi_1\rangle$ is only a difference in our knowledge, not in the real state, which could be the $\lambda$ in the overlap. 

Then PBR prepares two states (or $n$ in a more general situation) using different preparation devices for this. Each device creates, by the choice of the experimenter, $|\psi_0\rangle$ or $|\psi_1\rangle$. Then these states are combined, and such a measurement is used that for each of the possible outcomes of the measurement the outcome would be 0 for some particular combination of $|\psi_0\rangle$ resp. $|\psi_1\rangle$ of the different prepared states. Now, if the outcome would depend only on the state $\lambda$, and this $\lambda$ happens to be with some probability $q$ in the overlap, then with probability $q^2$ it would be completely unclear which of the four combinations $|00\rangle, |0+\rangle, |+0\rangle, |++\rangle$ would be the origin of this particular pair. Now, by construction for each of these four possibilities one of the four possible results would appear only with probability zero. But what happens depends only on the two $\lambda_1, \lambda_2$ which both do not have the information if their origin was $|0\rangle$ or $|+\rangle$. Whatever the result, we would obtain a non-zero probability for this particular result even if quantum theory predicts for the corresponding origins probability zero. 

But in the paleoclassical intepretation this does not work. The point is that it is not questioned at all that the states $|\psi_0\rangle$ resp. $|\psi_1\rangle$ have no overlap. Indeed, these states have been prepared by different preparation procedures, and the configurations of the preparation devices are part of the complete configuration, thus, part of the ontology of the paleoclassical interpretation.  

And this counterargument extends to quite arbitrary constructions of this type: Whenever we consider different wave functions, these wave functions have to be prepared somehow. The minimal interpretation defines only one way to prepare such a wave function: A corresponding measurement so that the wave function which is prepared is the eigenstate which corresponds to some eigenvalue of the measurement. The eigenvalue is, then, after the preparation part of the configuration of the universe, namely part of the configuration of the measurement device (the position of its pointer). 

The wave function which does, according to the paleoclassical interpretation, not contain any information about the configuration of the environment at all, would be the wave function of the whole universe. 

\subsection{Conclusion}

So the consideration of the arguments in favour of a beable status for the wave function, even if they seem to be strong and decisive at a first look, appear to be in no conflict at all with the interpretation of the wave function of a closed system in terms of incomplete information about this system.

The most important point to understand this is, of course, the very fact that the conditional wave function of the small subsystems of the universe we can consider really have a different character -- they depend on the configuration of the environment. And, so, the argumentation that these conditional wave functions describe real degrees of freedom, external to the system itself, is accepted and even derived from the interpretation.  

Nonetheless, the point that nor the complexity of the wave function of the whole universe, nor its evolution in time, nor the physical character of the laws of this evolution are in any conflict with an interpretation in terms of incomplete information is an important insight too. 

\section{The Wallstrom objection}  \label{Wallstrom}

Wallstrom \cite{Wallstrom} has made an objection against giving the fields of the polar decomposition $\rho(q)$, $S(q)$ (instead of the wave function $\psi(q)$ itself) a fundamental role.  

The first point is that around the zeros of the quantum mechanical wave function, the flow has no longer a potential. The quantum flow is a potential one only where $\rho(q)>0$. But in general there will be submanifolds of dimension $n-2$ where the wave function is zero. And for a closed path $q(s)$ around such a zero submanifold one finds that
\begin{equation}\label{curl}
\oint m_{ij} v^i(q) dq^j \neq 0.
\end{equation}

This, in itself, is unproblematic for the interpretation: The condition of potentiality is not assumed to be a fundamental one -- the fundamental object is not $S(q)$ but the $v^i(q)$. There will be some mechanism in subquantum theory which locally reduces violations of potentiality, so we can assume that the flow is a potential one only as an approximation. 

It is also quite natural to assume that such a mechanism works more efficient for higher densities and fails near the zeros of the density. 

So having them localized at the zeros of the density is quite nice -- not for a really fundamental equation, which is not supposed to have any infinities, but if we consider quantum theory as being only an approximation.  

The really problematic part of the Wallstom objection is a different one: It is that the quantum flow has to fulfill a  nontrivial \emph{quantization condition}, namely 
\begin{equation}\label{curlquantization}
\oint m_{ij} v^i(q) dq^j = \oint \pd_j S(q) dq^j = 2\pi m\hbar, \qquad m\in\Z.
\end{equation}
The point is, in particular, that the equations \eqref{Bohm}, \eqref{continuity} in flow variables are not sufficient to derive this quantization condition. So, in fact, this set of equations is \emph{not} empirically equivalent to the \Sch equation. 

This is, of course, no wonder, given the fact that the equivalence holds only for $\rho(q)=|\psi(q)|^2>0$. But, however natural, empirical inequivalence is empirical inequivalence. 

Then, this condition looks quite artificial in terms of the $v^i$. What is a triviality in terms of the wave function -- that it has to be uniquely defined globally -- looks extremely artificial and strange if formulated in terms of the $v^i(q)$. As Wallstrom \cite{Wallstrom} writes, to ``the best of my knowledge, this condition [\eqref{curlquantization}] has not yet found any convincing explanation outside the context of the \Sch equation''.  

\subsection{A solution for this problem} 

Fortunately I have found a solution for this problem in \cite{againstWallstrom}. I do not claim that it is a complete one -- there is a part which is beyond the scope of an interpretation, which has to be left to particular proposals for a subquantum theory. One has to check if the assumptions I have made about such a subquantum theory are really fulfilled in that particular theory.  

The first step of the solution is to recognize that, for empirical equivalence with quantum theory, it is sufficient to recover only solutions with simple zeros. Such simple zeros give $m=\pm 1$ in the quantization condition \eqref{curlquantization}. This is a consequence of the principles of general position: A small enough modification of the wave function cannot be excluded by observation, but leads to a wave function in general position, and in general position the zeros of the wave function are non-degenerated. 

The next step is a look at the actual solutions. For the simple, two dimensional, rotational invariant, zero potential case these solutions are defined by $S(q)=m\f$, $\rho(q)=r^{2|m|}$. And this extends to the general situation, where $S(q)=m\f+\tilde{S}(q)$,  $\rho(q)=r^{2|m|}\tilde{\rho}(q)$, such that $\tilde{S}(q)$ is well-defined in a whole environment of the zero, and $\tilde{\rho}(0)>0$.

But that means we can replace the problem of justifying an integer $m$ in $S(q)=m\f$, where all values of $m$ seem equally plausible, by the quite different problem of justifying $\rho(q)=r^2$ (once we need only $m=\pm 1$) in comparison with other $\rho(q)=r^\a$. This is already a quite different perspective. 

We make the natural conclusion and invent a criterion which prefers $\rho(q)=r^2$ in comparison with other $r^\a$.  This is quite easy:  

\begin{postulate}[regularity of $\Delta\rho$]\label{preg}
If $\rho(q)=0$, then $0< \Delta \rho(q) < \infty$ almost everywhere. 
\end{postulate}

This postulate already solves the inequivalence argument. The equations \eqref{Bohm}, \eqref{continuity} for $\rho(q)>0$, together with postulate \ref{preg}, already defines a theory empirically equivalent to quantum theory (even if the equivalence is not exact, because only solutions in general position are recovered). 

It remains to invent a justification for this postulate. 

The next step is to rewrite equation \eqref{Bohm} for stable states in form of a balance of energy densities. In particular, we can rewrite the densitized quantum potential as
\begin{equation}\label{Q}
Q[\rho]\rho = \frac12\rho u^2 - \frac14 \Delta \rho,
\end{equation}
with the ``osmotic velocity'' $u(q) = \frac12 \nabla \ln \rho(q)$. Then the energy balance looks like
\begin{equation}
\frac12\rho v^2 + \frac12 \rho u^2 + V(q) = \frac14\Delta\rho.
\end{equation}

So, the operator we have used in the postulate is not an arbitrary expression, but a meaningful term, which appears in an important equations -- an energy balance.  This observation is already sufficient to justify the $\Delta \rho(q) < \infty$ part of the condition. There may be, of course, subquantum theories which allow for infinities in energy densities, but it is hardly a problem for a subquantum theory to justify that expressions which appear like energy densities in energy balances have to be finite.  

Last but not least, subquantum theory has to allow for nonzero curl $\w$, but has to suppress it to obtain a quantum limit. One way to suppress it is to add a penalty term $U(\rho,\w)$ which increases with $|\w|$. This would give
\begin{equation}
\frac12\rho v^2 + \frac12 \rho u^2 + V(q) + U(\rho,\w) = \frac14\Delta\rho.
\end{equation}
Moreover subquantum theory has to regularize the infinities of $v$ and $u$ at the zeros of the density. One can plausibly expect that this gives finite but large values of $|\w|$ at zero which decrease sufficiently fast with $r$. Now, a look at the energy balance shows that, if the classical potential term $V(q)$ is neglected (for example assuming that it changes only smoothly) the only term which can balance $U(\rho,\w)$ at zero is the $\Delta \rho$ terms, which, therefore, has to be finite but nonzero. Or, at least, it would not be difficult to modify the definition of $U(\rho,\w)$ in such a way that the extremal value $\Delta\rho=0$ (we have necessarily $\Delta\rho\ge 0$ at the minima) has to be excluded. 

So the postulate seems nicely justifiable. For some more details I refer to \cite{againstWallstrom}. What remains is the particular job of particular proposals for subquantum theories -- they have to check if the way to justify the postulate really works in this theory, or if it may be justified in some other way. But this is beyond the scope of the interpretation. What has to be done by the interpretation -- in particular, to obtain empirical equivalence with quantum theory -- has been done.
 
\section{A Popperian argument for preference of an information-based interpretation}

One of Popper's basic ideas was that we should prefer -- as long as possible without conflict with experience -- theories which are more restrictive, make more certain predictions, depend on less parameters.  And, while this criterion has been formulated for theories, it should be applied, for the same reasons, to more general principles of constructing theories too. 

This gives an argument in preference for an interpretation in terms of incomplete information. 

Indeed, let's consider, from this point of view, the difference between interpretations of fields $\rho(q)$, $v^i(q)$ in terms of a probability for some real trajectories $q(t)$, and interpretations which reify them as describing some external reality, different from $q(t)$, which influences the trajectory $q(t)$.  

It is quite clear and obvious which of the two approaches is more restrictive. Designing theories of the first type, we are restricted, for the real physics, to theories for single trajectories $q(t)$. Then, given that we have identified the connection between the fields $\rho(q)$, $v^i(q)$ and $q(t)$ as those of a probability flow, everything else follows. There is no longer any freedom of choice for the field equations. If we have fixed the Hamiltonian evolution for $p(t), q(t)$, the Liouville field equation for $\rho(p,q)$ is simply a logical consequence. Similarly, the continuity equation \eqref{continuity} is a law of logic, it cannot be modified, is no longer a subject of theoretical speculation. It is fixed by the interpretation of $\rho(q)$, $v^i(q)$ as a probability flow, in the same way as $\rho(q)\ge 0$ is fixed. 

In the second case, we have much more freedom -- the full freedom of speculation about field theories in general. In particular, the continuity equation can be modified, introducing, say, some creation and destruction processes, which are quite natural if $\rho(q)$ describes a density of some external objects. 

The derivation of the $U(1)$ global phase shift symmetry is another particular example of such an additional logical law following from the interpretation.  

So there are some consequences of the interpretation which have purely logical character, including the continuity equation, $\rho(q)\ge 0$, and the $U(1)$ global phase shift symmetry.  But these will not be the only consequences. The other equations will be restricted, in comparison with field theories, too, but in a less clear and obvious way. There is, last but not least, a large freedom of choice for the equations of the real beables $q(t)$, which corresponds to a similarly large freedom of choice of the resulting equations for $\rho(q)$, $v^i(q)$. But this freedom of choice will be, nonetheless, much smaller than the completely arbitrariness of a general field theory.  

This consideration strongly indicates that we have to prefer the interpretation in terms of incomplete information until it has been falsified, until it appears incompatible with observation. 

The immediate, sufficiently trivial logical consequences we have found yet are compatible with the \Sch equation and therefore with observation. So we should prefer this interpretation. 

\section{Open problems}

Instead of using such a Popperian argumentation, I could have, as well, used simply Ockham's razor: Don't multiply entities without necessity. Once, given this interpretation, there is no necessity for more than a single classical trajectory $q(t)$, one should not introduce other real entities like really existing wave functions. 

But the Popperian consideration has the advantage that it implicitly defines an interesting research program.

\subsection{Other restrictions following from the interpretation}\label{viability}

In fact, given the restrictive character of the interpretation, there may be other, additional, more subtle restrictions of the equations for probability flows $\rho(q)$, $v^i(q)$, restrictions which we have not yet identified, but which plausibly exist. 

So what are these additional restrictions for equations for probability flows $\rho(q)$, $v^i(q)$ in comparison with four general, unspecific fields fulfilling a continuity equation? I have no answer. 

This is clearly an interesting question for future research. It is certainly also a question interesting in itself, interesting from point of view of pure mathematics, for a better understanding of probability theory.  

The consequences may be fatal for this approach -- it may be that we find that the \Sch equation does not fit into this set of restrictions even approximately\footnote{That it does not fit exactly will be shown below, in the consideration about quantum computers. But for the interpretation to be viable, exactness is not obligatory -- it is sufficient that the \Sch equation can be obtained as an approximation.}. This possibility of falsification is, of course, the very point of the Popperian consideration. I'm nonetheless not afraid that this happens, but this is only a personal opinion. 

The situation may be, indeed, much better: That this subclass of theories contains the \Sch equation, but appears to be heavily restricted by some additional, yet unknown, conditions. Then, all of these additional restrictions would give us additional partial answers to the ``why the quantum'' question. 

\subsection{Why is the \Sch equation linear?}

The most interesting question which remains open is why the \Sch equation is linear. We have found only some part of the answer -- an explanation for the global $U(1)$ phase symmetry based on the informational content and of the homogeneity based on the reduction of the equation to subsystems. 

But, given that the pre-\Sch equation is non-linear, but interpreted in the same way, the linearity of the \Sch equation cannot follow from the interpretation taken alone. Some other considerations are necessary to obtain linearity.

One idea is to justify linearity as an approximation. Last but not least, linearization is a standard way to obtain approximations. 

The problem of stability in time may be also relevant here. The pre-\Sch equation becomes invalid after a short period of time, when the first caustic appears. There is no such problem in quantum theory, which has a lot of stable solutions. Then, there should be not only stable solutions, but also slowly changing solutions: It doesn't even matter if the fundamental time scale is Planck time or something much larger -- even if it is only the time scale of strong interactions, all the things changing around us are changing in an extremely slow way in comparison with this fundamental time scale. But the linear character of the \Sch equation gives us a way to obtain solutions slowly changing in time by combining different stable solutions with close energies. 

\section{A theoretical possibility to test: The speedup of quantum computers}\label{computer}

There is an idea which suggests, at least in principle, a way to distinguish observationally the paleoclassical interpretation (or, more accurate, the class of all more fundamental theories compatible with the paleoclassical interpretation) from the minimal interpretation. 

The idea is connected with the theory of quantum computers. If quantum computers really work as predicted by quantum theory, these abilities will provide fascinating tests of the accuracy of quantum theory. In the case of Simon's algorithm, the speed-up is exponential over any classical algorithm. It may be a key for the explanation of this speed-up that the state space (phase space) of a composite classical system is the Cartesian product of the state spaces of its subsystems, while the state space of a composite quantum system is the tensor product of the state spaces of its subsystems. For $n$ qubits, the quantum state space has $2^n$ instead of $n$ dimensions. So the information required to represent a general state increases exponentially with $n$ (see, for example, \cite{Bub}). There is also the idea ``that a quantum computation is something like a massively parallel classical computation, for all possible values of a function. This appears to be Deutsch's view, with the parallel computations taking place in parallel universes.'' \cite{Bub}. 

It is this \emph{exponential} speedup which suggests that the predictions of standard QM may differ from those of the paleoclassical interpretation. An exact quantum computer would have all beables contained in the wave function as degrees of freedom.  A quantum computer in the paleoclassical interpretation has only the resources provided by its beables.  But these beables are, essentially, only the classical states of the universe. Given the exponential difference between them, $n$ vs. $2^n$ dimensions for qubits instead of classical bits, an exact quantum computer realizable at least in principle in a laboratory on Earth can have more computational resources than a the corresponding computer of the paleoclassical interpretation, which can use only the classical degrees of freedom, even if these are the classical degrees of freedom of the whole observable universe. 

But if we distort a quantum computer, even slightly, the result will be fatal for the computation. In particular, if this distortion is of the type of the paleoclassical interpretation, which replaces an exact computer with a $2^n$-dimensional state space by an approximate one with only $N$ dimensions, then even for quite large $N\gg n$ the approximate computer will be simply unable to do the exact computations, even in principle. There simply are no parallel universes in the paleoclassical interpretation to make the necessary parallel computations. 

So, roughly speaking, the prediction of the paleoclassical interpretation is that a sufficiently large quantum computer will fail to give the promised exponential speedup. The exponential speedup will work only up to a certain limit, defined by the logarithm of the relation between the size of the whole universe and the size of the quantum computer. 

Of course, we do not know the size of the universe. It may be much larger than the size of the observable universe, or even infinite. Nonetheless, this argument, applied to any finite model of the universe, shows that the true theory, the theory in the configurational beables alone, cannot be exactly quantum theory. This is in my opinion the most interesting conclusion. 

But let's see if we can, nonetheless, make even testable (at least in principle) predictions. So let's presuppose that the universe is finite, and, moreover, let's assume that this size is not too many orders larger than the its observable part. This would be already sufficient to obtain some numbers about the number of qubits such that the $2^n$ exponential speedup is no longer possible. This number will be sufficiently small, small enough that a quantum computer in a laboratory on Earth will be sufficient to reach this limit. 

And, given the logarithmic dependence on $N$, increasing $N$ does not help much. If it is possible to build a quantum computer with $n$ qubits, why not with $2n$? This would already move the size of the universe into completely implausible regions. 

\subsection{The speed of quantum information as another boundary}

Instead of caring about the size of the universe, it may be more reasonable to care about the size of the region which can causally influence us. Here I do not have in mind the limits given by of relativity, by the speed of light. Given the violation of Bell's inequality, there has to be (from a realist's point of view) a hidden preferred frame where some other sort of information -- quantum information -- is transferred with a speed much larger than the speed of light.  But if we assume that the true theory has some locality properties, even if only in terms of a much larger maximal speed, the region which may be used by a quantum computer for its computational speedup decreases in comparison with the size of the universe. 

So if we assume that there is such a speed limit for quantum information, then we obtain in the paleoclassical interpretation even more restrictive limits for the speedup reachable by quantum computers, limits which depend logarithmically on the speed limit for quantum information. 

Nonetheless, I personally don't believe that quantum computers will really reach large speedups. I think the general inaccuracy of human devices will prevent us from constructing quantum computers which can really use the full $2^n$ power for large enough $n$. I would guess that the accuracy requirements necessary to obtain a full $2^n$ speedup will also grow exponentially. So I guess that quantum computers will fail already on a much smaller scale. 

\section{Conclusions}

So it's time to summarize:

\begin{itemize}

\item The unknown, true theory of the whole universe is a theory defined on the classical configuration space $Q$, with the configuration $q(t)$ evolving in absolute time $t$ as a complete description of all beables. 

\item The wave function of the whole universe is interpreted as a consistent set of incomplete information about these fundamental beables. 

\item In particular, $\rho(q)$ defines not some ``objective'' probability, but an incomplete set of information about the real position $q$, described, as required by the logic of plausible reasoning, by a probability distribution $\rho(q)dq$.  

\item The phase $S(q)$ describes, via the ``guiding equation'', the expectation value $\langle\dot{q}\rangle$ of the velocity given the actual configuration $q$ itself. So the ``guiding equation'' is not a physical equation, but has to be interpreted as part of the definition of $S(q)$, which describes which information about $q$ is contained in $S(q)$.    

\item Only a constant phase factor of $\psi(q)$ does not contain any relevant information about the trajectory $q(t)$. Therefore, the equations for $\psi(q)$ should not depend on such a factor. 

\item The \Sch equation is interpreted as an approximate equation. More is not to be expected, given that it describes the evolution of an incomplete set of information. 

\item The linear character of the \Sch equation is interpreted as an additional hint that it is only an approximate equation. 

\item The interpretation can be used to reinterpret Nelsonian stochastics. The resulting picture is conceptually more consistent than the original proposal. 

\item The Wallstrom objection appears much less serious than expected. The quantization condition for simple zeros (which is sufficient because it is the general position) can be derived from the much simpler regularity postulate that $0<\Delta \rho(q)<\infty$ if $\rho(q)=0$. While a final justification of this condition has to be left to a more fundamental theory, it is, as shown in \cite{againstWallstrom}, plausible that this is not a problem for such theories. 

\item If the true theory of the universe is defined on classical configurations, and the whole universe is finite, quantum computers can give their promised exponential speedup only up to an upper bound for the number of qubits, which is much smaller than the available qubits of the universe. This argument shows that the \Sch equation has to be approximate. 

\item The dBB problem with of the ``action without reaction'' asymmetry is solved: For effective wave function, the collapse defines the back-reaction, for the wave function of the whole universe there should be no such back-reaction -- it is only an equation about incomplete information about reality, not about reality itself. 

\item The wave functions of small subsystems obtain a seemingly objective, physical character only because they, as conditional wave functions, depend on the physical beables of the environment. 

\end{itemize}

From point of view of simplicity, the paleoclassical interpretation is superior to all alternatives. The identification of the fundamental beables with the classical configuration space trajectory $q(t)$ is sufficient for this point. 

It has also the additional advantage that it leads to strong restrictions of the properties of a more fundamental, sub-quantum theory:  It has to be a theory completely defined on the classical configuration space. Moreover, it has to be a theory which, in its statistical variant, leads to a Fokker-Planck-like equations for the probability flow defined by the classical flow variables $\rho(q)$ and $v^i(q)$. 

\begin{appendix}
 
\section{Compatibility with relativity}\label{relativity}

Most physicists consider the problem of compatibility with relativity as the major problem of dBB-like interpretations -- sufficient to reject them completely. But I have different, completely independent reasons for accepting a preferred frame, so that I don't worry about this. 

There are two parts of this compatibility problem, a physical and a philosophical one, which should not be mingled: 

The physical part is that we need a dBB version of relativistic quantum field theories, in particular of the standard model of particle physics -- versions which do not have to change the fundamental scheme of dBB, and, therefore, may have a hidden preferred frame. 

The philosophical part is the incompatibility of a hidden preferred frame with relativistic metaphysics.  

The physical problem is heavily overestimated, in part because of the way dBB theory is often presented: As a theory of many particles. The appropriate way is to present it as a general theory in terms of an abstract configuration space $Q$, and to recognize that field theories as well as their lattice regularizations fit into this scheme.  The fields are, of course, fields on three-dimensional space $\R^3$ changing in time,  and their lattice regularizations live on three-dimensional spatial lattices $\Z^3$, not four-dimensional space-time lattices. But this violation of manifest relativistic symmetry is already part of the second, philosophical problem. 

The simple, seemingly non-relativistic Hamiltonian \eqref{HamiltonFunction}, with $p^2$ instead of $\sqrt{p^2 + m^2}$, is also misleading: For relativistic field theories the quadratic Hamiltonian is completely sufficient. Indeed, a relativistic field Lagrangian is of type
\begin{equation}\label{field}
\mathscr{L} = \frac{1}{2}((\pd_t\varphi)^2 - (\pd_i\varphi)^2)-V(\varphi).
\end{equation}
This gives momentum fields $\pi = \pd_t\varphi$ and the Hamiltonian
\begin{equation}
\mathscr{H} = \frac{1}{2}(\pi^2 + (\pd_i\varphi)^2)+V(\varphi) = \frac{1}{2}\pi^2 + \tilde{V}(\varphi)
\end{equation}
quadratic in $\pi$, thus, the straightforward field generalization of the standard Hamiltonian \eqref{HamiltonFunction}. And for a lattice regularization, the Hamiltonian is already exactly of the form \eqref{HamiltonFunction}. So, whatever one thinks about the dBB problems with other relativistic fields, it is certainly not relativity itself which causes the problem. 

The problem with fermions and gauge fields is certainly more subtle. Here, my proposal is described in \cite{clm}. It heavily depends on a preferred frame, but for completely different reasons -- interpretations of quantum theory are not even mentioned. Nonetheless, fermion fields are obtained from field theories of type \eqref{field}, and gauge-equivalent states  are interpreted as fundamentally different beables, so that no BRST factorization procedure is necessary. 

Another part of the physical problem is compatibility with relativistic gravity. Here I argue that it is the general-relativistic concept of background-freedom which is incompatible with quantum theory and has to be given up. I use a quantum variant of the classical hole argument for this purpose \cite{hole}. As a replacement, I propose a theory of gravity with background and preferred frame \cite{glet}.  

So there remains only the philosophical part. But here the violation of Bell's inequality gives a strong argument in favour of a preferred frame: Every realistic interpretation needs it. Moreover, the notion of metaphysical realism presupposed by ``realistic interpretation'' is so weak that Norsen \cite{Norsen} has titled a paper ``against realism'', arguing that one should not mention realism at all in this context. The metaphysical notion of realism used there is so weak that to give it up does not save Einstein locality at all -- it is presupposed in this notion too. 

\section{Pauli's symmetry argument}\label{Pauli}

There is also another symmetry argument against dBB theory, which goes back to Pauli \cite{Pauli}, which deserves to be mentioned: 

\begin{quote}
\ldots the artificial asymmetry introduced in the treatment of the two variables of a canonically conjugated pair
characterizes this form of theory as artificial metaphysics. (\cite{Pauli}, as quoted by \cite{Freire}),

``\ldots the Bohmian corpuscle picks out by fiat a preferred basis (position) \ldots'' \cite{BrownWallace}
\end{quote}

Here my counterargument is presented in \cite{kdv}. I construct there an explicit counterexample, based on the KdV equation,  that the Hamilton operator alone, without a specification which operator measures position, is not sufficient to fix the physics. It follows that the canonical operators have to be part of the complete definition of a quantum theory and so have to be distinguished by the interpretation as something special, different from the other similar pairs of operators. 

The Copenhagen interpretation makes such a difference -- this is one of the roles played by the classical part. But attempts to get rid of the classical part of the Copenhagen interpretation, without adding something else as a replacement, are not viable \cite{pure}. One has to introduce a replacement. 

Recognizing that the configuration space has to be part of the definition of the physics gives more power to an old argument in favour of the pilot wave approach, made already by de Broglie at the Solvay conference 1927:
\begin{quote}
``It seems a little paradoxical to construct a configuration space with the coordinates of points which do not exist.'' \cite{deBroglie}.
\end{quote}

\section{Problems with field theories}\label{fields}

It has been argued that fields are problematic as beables in general for dBB theory, a point which could be problematic for the paleoclassical interpretation too. 

In particular, the equivalence proof between quantum theory and dBB theory depends on the fact that the overlap of the wave function for different macroscopic states is irrelevant. But it appeared in field theory that for one-particle states there is always a non-trivial overlap, even if these field states are localized far away from each other.  

But, as I have shown in \cite{overlap}, the overlap decreases sufficiently fast (approximately exponentially) with greater particle numbers. 

\section{Why we observe configurations, not wave packets}\label{wavepackets}

In the many worlds community there is a quite popular argument against dBB theory -- that it is many worlds in denial (for example, see \cite{BrownWallace}). But this argument depends on the property of dBB theory that the wave function is a beable, a really existing object. So it cannot be applied against the paleoclassical interpretation, where the wave function is no longer a beable. 

But in fact it is invalid also as an argument against dBB theory. In fact, already in dBB theory it is the configuration $q(t)$ which is observable and not the wave function.  

This fact is sometimes not presented in a clear enough way, so that misrepresentations become possible. For example Brown and Wallace \cite{BrownWallace} find support for another interpretation even in Bohm's original paper \cite{Bohm}:
\begin{quote}
\ldots even in his hidden variables paper II of 1952, Bohm seems to associate the wavepacket chosen by the corpuscles as the representing outcome of the measurement -- the role of the corpuscles merely being to point to it.
(\cite{BrownWallace} p. 15)\end{quote}
and support their claim with the following quote from Bohm
\begin{quote}
Now, the packet entered by the apparatus variable $y$ determines the actual result of the measurement, which the observer will obtain when he looks at the apparatus.
(\cite{Bohm} p. 118)\end{quote}
This quote may, indeed, lead to misunderstandings about this issue. So, maybe we observe only the wave packet containing the configuration, instead of configuration itself? 

My answer is a clear no. I don't believe into the existence of sufficiently localized wave packets to construct some effective reality out of them, as assumed by many worlders. 

Today they use decoherence to justify their belief that wave packets will be sufficiently localized. But decoherence presupposes another structure -- a decomposition of the world into systems. Only from point of view of such a subdivision of $q$ into, say, $(x,y)$, a non-localized wave function like $e^{-(x-y)^2/2}$ may look similar to a superposition, for different $a$, of product states localized in $x$ and $y$ like $e^{-(x-a)^2}\cdot e^{-(y-a)^2}$. 


But where does this subdivision into systems come from?  The systems around us -- observers, planets, measurement devices -- cannot be used for this purpose. They do not exist on the fundamental level, as a predefined structure on the configuration space. But the subdivision into systems has to, once we need it to construct localized objects. Else, the whole construction would be circular.

So one would have to postulate something else as a fundamental subdivision into systems. This something else is undefined in the interpretations considered here, so an interpretation based on it is simply another interpretation, with another, additional fundamental structure -- a fundamental subdivision into systems. 

\end{appendix}

\end{document}